%% file: main.tex
\documentclass[a4paper,UKenglish,cleveref, autoref, thm-restate]{lipics-v2021}

\usepackage[T1]{fontenc}
\usepackage{lmodern, tikzpeople}
\usetikzlibrary{positioning, shapes.multipart}
\usepackage{mathtools}
\usepackage{cite}
\usepackage{graphicx}
\usepackage{caption}
\usepackage{tikz}
\usepackage{enumitem}
\usepackage{url}

\tikzset{%
  Line Width/.code={%
    \pgfpointxy{#1}{0}%
    \pgfgetlastxy\tmpx\tmpy\tikzset{line width/.expanded=\tmpx}%
  },
  person/.style={
    Line Width=1/4, draw=black, color=gray!50
  },
  highlight/.style={
    preaction={Line Width=1/4, draw=gray!50, path fading=west, fading angle=-45},
    Line Width=1/4, draw=white, , path fading=east, fading angle=-45,
}}
\tikzset{>=latex}

\usepackage{amsthm}

\usepackage{algpseudocode}
\input{imports/macros.tex}
\input{imports/notations.tex}

\usepackage{todonotes}

\newcommand{\tx}{\mathit{tx}}
\newcommand{\tsig}{\mathit{tsig}}
\newcommand{\nonce}{\mathit{nonce}}
\newcommand{\hash}{\mathit{hash}}

\DeclarePairedDelimiter\ceil{\lceil}{\rceil}

\def\BibTeX{{\rm B\kern-.05em{\sc i\kern-.025em b}\kern-.08em
    T\kern-.1667em\lower.7ex\hbox{E}\kern-.125emX}}

\title{A Fair and Resilient Decentralized Clock Network for Transaction Ordering}

\author{Andrei Constantinescu}{ETH Z{\"u}rich, Switzerland}{aconstantine@ethz.ch}{https://orcid.org/0009-0005-1708-9376}{}

\author{Diana Ghinea}{ETH Z{\"u}rich, Switzerland}{ghinead@ethz.ch}{https://orcid.org/0000-0002-5294-9459}{}

\author{Lioba Heimbach}{ETH Z{\"u}rich, Switzerland}{hlioba@ethz.ch}{https://orcid.org/0000-0002-8258-1712}{}

\author{Zilin Wang}{ETH Z{\"u}rich, Switzerland}{ziwang@ethz.ch}{}{}

\author{Roger Wattenhofer}{ETH Z{\"u}rich, Switzerland}{wattenhofer@ethz.ch}{https://orcid.org/0000-0002-6339-3134}{}

\authorrunning{A. Constantinescu, D. Ghinea, L. Heimbach, Z. Wang, and R. Wattenhofer} 

\Copyright{Andrei Constantinescu, Diana Ghinea, Lioba Heimbach, Zilin Wang, and Roger Wattenhofer} 



\category{} 

\nolinenumbers 

\EventEditors{Alysson Bessani, Xavier D\'efago, Yukiko Yamauchi, and Junya Nakamura}
\EventNoEds{4}
\EventLongTitle{27th International Conference on Principles of Distributed Systems (OPODIS 2023)}
\EventShortTitle{OPODIS 2023}
\EventAcronym{OPODIS}
\EventYear{2023}
\EventDate{December 6--8, 2023}
\EventLocation{Tokyo, Japan}
\EventLogo{}
\SeriesVolume{286}
\ArticleNo{5}

\begin{document}

\maketitle 
\begin{abstract}
Traditional blockchain design gives miners or validators full control over transaction ordering, i.e.,~they can freely choose which transactions to include or exclude, as well as in which order. While not an issue initially, the emergence of decentralized finance has introduced new transaction order dependencies allowing parties in control of the ordering to make a profit by front-running others' transactions.
In this work, we present the \textit{Decentralized Clock Network}, a new approach for achieving fair transaction ordering. 
Users submit their transactions to the network's clocks, which run an agreement protocol that provides each transaction with a timestamp of receipt which is then used to define the transactions' order.
By separating agreement from ordering, our protocol is efficient and has a simpler design compared to other available solutions. Moreover, our protocol brings to the blockchain world the paradigm of asynchronous fallback, where the algorithm operates with stronger fairness guarantees during periods of synchronous use, switching to an asynchronous mode only during times of increased network delay. 
\end{abstract}

\section{Introduction}

The first blockchain, a decentralized distributed digital ledger that records transactions across a network of computers, was introduced in 2008 with Bitcoin by Nakamoto~\cite{nakamoto2008bitcoin}. Blockchains offer a novel way of storing and transferring value in a trustless and secure manner, without the need for intermediaries. Despite their popularity, blockchain adoption was slow, as blockchains were, initially, mainly used to facilitate simple transfers of money between two individuals. However, this changed in 2015 with the introduction of smart contracts on Ethereum~\cite{wood2014ethereum}, allowing for complex digital agreements to be carried out on-chain. Nowadays, smart contracts are the backbone of a rapidly-growing complex ecosystem of decentralized financial applications known as \emph{decentralized finance (DeFi)}. DeFi offers most traditional financial services, including decentralized exchanges, lending protocols, and stablecoins, without relying on financial intermediaries.

The smart contracts that govern DeFi are generally dependent on the transaction order. That is, the outcome of executing a set of transactions depends on their order. As most transactions were simple transfers in the early days, the original blockchain design did not need to pay much attention to transaction ordering. Instead, the power of transaction ordering is concentrated in miners or validators, which can freely choose which transactions to include and how to order them inside each block. Nowadays, block proposers (miners) extract profit from appropriately ordering, including, and excluding transactions during block production. This profit is known as miner (or maximal) extractable value (MEV). MEV accounts for a profit of at least US\$ 650M ~\cite{mevexplore} so far. In fact, Flashbots and other transaction relay protocols organized a whole market around ordering transactions.

\paragraph*{Front-running Attacks}
Most MEV relies on the ability of the attacker to \emph{front-run} the victim's transaction $\tx.$ To be specific, the attacker observes a newly generated victim transaction $\tx$ in the mempool (the public waiting area for transactions). The attacker then introduces their own transaction $\tx'.$ If $\tx'$ executes before $\tx$ (front-running), the attacker profits at the expense of the victim. So, the attacker may simply bribe the block proposer with a high fee to execute $\tx'$ first, even though $\tx'$ was only created once $\tx$ was already publicly known.

Front-running can be broadly categorized into two types~\cite{qin2022quantifying}:~\emph{tolerant} and \emph{destructive}. Tolerant front-running involves the attacker placing their own transaction before the victim's transaction in the order of execution. This allows the attacker to gain an advantage, such as purchasing an asset at a lower price before the victim can. Such attacks are often seen on decentralized exchanges, where the attacker executes a trade before the victim, reaping the benefits of price changes. Destructive front-running, on the other hand, has the attacker taking out the victim's transaction altogether. Generally, the attacker copies the victim's presumably profitable transaction. If the attacker's transaction executes first, the victim's transaction would no longer execute, at least not as intended.

\paragraph*{Our Contribution} 
We propose the Decentralized Clock Network (DCN), a novel solution for achieving fair transaction ordering. More concretely, our system ensures that, if a transaction $\tx$ was sent to the system long enough before transaction $\tx',$ then $\tx'$ cannot be ordered before $\tx,$ i.e., preventing tolerant front-running. 
In contrast to most previous solutions relying on the blockchain consensus algorithm to determine a relative ordering of the transactions, our approach employs a decentralized network of $n$ nodes, equipped with clocks, resilient to $f < n / 3$ 
byzantine failures to agree on a timestamp for each transaction. These timestamps are subsequently used to determine the order of the transactions inside each block and across blocks. Decoupling timestamping from ordering enables lower latency bounds whilst reducing the complexity of the consensus mechanism.

A blockchain system is synchronous if all messages arrive at the receiver within a known time-bound, and the nodes involved have local clocks that are (almost) perfectly synchronized.
However, in times of turmoil, such as when participants are under attack, messages might experience longer delays, or clocks may no longer be aligned with real-time. Such failures are modeled by the asynchronous model. An important novelty in our work is that our protocol is designed to provide guarantees regardless of the network conditions, without knowing in advance which setup to expect. It is designed for the asynchronous model, however, if the network happens to satisfy some synchrony assumptions, which is often the case in real-world networks, it provides stronger guarantees reflecting in the order obtained. To quantify this effect, we propose a new notion of order fairness, called $\delta$-Median Fairness. Roughly, transactions shall be ordered based on a value that is close to the median of the points in time when honest nodes in the DCN first learn about the transaction. Here, $\delta$ is an error parameter, determining the closeness of the estimated median to the true median of the honest timestamps in terms of quantiles.
This definition is a stronger version of Honest-Range Fairness (or fair separability, as defined in \cite{zhang2020byzantine}). When operating under asynchronous conditions, our algorithm achieves $f$-Median Fairness, which coincides with Honest-Range Fairness in the worst case $n = 3f + 1,$ but is stronger otherwise. On the other hand, when the network is synchronous for a sufficient amount of time, our algorithm achieves the superior guarantee of $\ceil*{f / 2}$-Median Fairness. In both cases, these guarantees are optimal.
We add that our protocol sidesteps the attack where relative orders relying on the median can be manipulated by a single byzantine node presented in~\cite{kelkar2020order} by ensuring that (1) nodes always agree on some honest timestamp, and (2) with the help of cryptographic primitives, we do not allow nodes, or anyone else, to see the transaction contents before a timestamp is agreed upon.

\paragraph*{Related Work}

\noindent \textbf{Fair Ordering.}
Blockchain front-running prevention techniques have been the subject of significant research in recent years. We point the reader to Baum et al.~\cite{baum2021sok} and Heimbach et al.~\cite{Heimbach2022sok} for an overview of these approaches and only discuss the most relevant in the following.

Flashbots~\cite{FlashBots} and other private relay services, in which transactions are sent directly to a trusted third party for ordering and subsequent forwarding to validators for block inclusion, are widely adopted. While this approach is efficient, it centralizes the transaction ordering process, i.e., introduces a single point of failure, and is often used to front-run as opposed to protect against. In contrast, our approach distributes the transaction ordering responsibility.

In the field of fair transaction ordering, committee-based approaches have been widely studied. Generally, these approaches can be divided into two categories:~those that can operate in asynchrony and those that assume partial synchrony, which is a model weaker than synchrony and stronger than asynchrony. To tackle fair ordering in partial synchrony, Pompe is proposed by Zhang et al.~\cite{zhang2020byzantine}, Wendy is proposed by Kursawe~\cite{kursawe2020wendy} and Themis is proposed by Kelkar et al.~\cite{kelkar2021themis}. As opposed to these protocols, the DCN we propose is equipped to handle asynchrony. In particular, Pompe and Themis rely on (partial) synchrony and Wendy assumes the clocks of the nodes are always synchronized. 

Kelkar et al.~\cite{kelkar2020order} introduce Aequitas, which achieves state-of-the-art fairness properties, but has a significant communication complexity of $\mathcal{O}(n^4)$ in asynchrony. Our agreement protocol achieves in expectation $\mathcal{O}(n^3 \log \Delta)$ message complexity in asynchrony, where $\Delta$ denotes the \emph{observed} network delay. We note that this delay does not have to be known a priori, as opposed to classical synchronous protocols.

Quick order fairness, introduced by Cachin et al.~\cite{cachin2022quick} achieves $\mathcal{O}(n^3)$ message complexity in asynchrony. While their protocol allows for a node to gain insider information before an ordering is agreed upon, our protocol adds further protection to users as the committee only sees the full transaction after the timestamp is agreed upon. Further, their approach, and the others, only target agreement amongst the permissioned committee, while our design extends to implementing the fair ordering on a permissionless blockchain after agreement has been reached in the permissioned committee.

\noindent \textbf{Agreement Protocols.}
Achieving agreement on a value subject to some Validity condition, i.e., Byzantine Agreement (BA) \cite{byzgenerals}, is an extensively studied problem in Distributed Computing. In real-world applications, hence also in our setting, it is desirable to expect the Validity condition to carry some meaning, while the classical BA definition only ensures that if honest nodes have the same input value $v,$ they all output $v.$ If this pre-agreement condition is not met, the honest nodes may output an adversarially chosen value. Recent works have focused on achieving more meaningful guarantees, such as ensuring that the honest output is \emph{close} to the honest inputs' median \cite{stolz2016byzantine}, to the $k$-th lowest honest input \cite{melnyk2018byzantine}, or somewhere in the range of honest inputs \cite{vaidya2013byzantine}. These works, however, only focus on the synchronous model. That is, they assume perfectly synchronized clocks and a publicly available upper bound on the network delay. A more realistic setting is the so-called asynchronous model, which drops this assumption, but showcases important limitations:~in the asynchronous setting, BA cannot be achieved deterministically \cite{FLP}.
There is still hope, however:~randomized asynchronous BA protocols exist
\cite{mostefaoui2015signature,rabin1993randomized,berman1993randomized,friedman2005simple,bracha1987asynchronous,toueg1984randomized,canetti93fast,cachin2000random}; however, without meaningful Validity guarantees if the input space contains more than two values. Another relaxed variant of BA is Approximate Agreement (AA) \cite{JACM:DLPSW86, OPODIS:AAD04}, which offers deterministic protocols that enable honest nodes to output values within the range of their inputs, with the caveat of weakening the Agreement guarantees:~honest outputs are $\varepsilon$-close for any predefined $\varepsilon > 0.$

To implement our fair-ordering definition, we propose an asynchronous (randomized) BA protocol with optimal resilience, that achieves Median Validity \cite{stolz_byzantine_nodate, melnyk2018byzantine} with optimal-error guarantees, assuming that the inputs are integers. Our lower bound on this error implies that, when the network is asynchronous, and when aiming for optimal resilience, the best one can hope for is obtaining outputs within the range of the honest inputs. We circumvent this problem by designing a protocol
\textcolor{black}{ whose Validity guarantees scale with the network conditions: if the synchrony assumptions are satisfied for a sufficient amount of time, our protocol will enable honest nodes to agree on a value satisfying the synchronous model's optimal guarantees on Median Validity. Otherwise, our protocol will at least provide Median Validity with optimal guarantees for the asynchronous model, hence the output agreed upon will be within the range of honest values.} Designing protocols that achieve simultaneously optimal guarantees in both synchronous and asynchronous networks, has been a topic that attracted increased attention in the recent years in the Distributed Computing literature. There has been a line of works focusing on problems such as Byzantine Agreement~\cite{TCC:BluKatLos19}, Approximate Agreement~\cite{PODC:GhLiWa22}, State Machine Replication~\cite{AC:BluKatLos21}, and also Multi-Party Computation~\cite{Crypto:BLL20,TCC:DHL21,PODC:ApChCh22}.

\section{The Decentralized Clock Network}

In this section, we describe the DCN, which consists of a network of nodes equipped with {synchronized} clocks operating with the objective of providing an accurate and decentralized timestamping service to blockchain transactions. The resulting timestamps are used to determine the ordering of the transactions inside each block, as well as across blocks. The intuition behind using a timestamping service is that, instead of relying on consensus to determine the ordering directly, like in FSS from ChainLink Labs~\cite{FSS}, this way the order of the transactions is naturally induced by the timestamps, allowing the complexity of the agreement protocol to be reduced.

\paragraph*{High-Level Design}
To enable DCN support for ordering transactions on an existing blockchain, the blockchain requires only minor adaptations. In particular, with every submitted transaction, an additional timestamp computed by the DCN is expected. Validators should check for each block whether timestamps are authentic and whether the ordering induced by the timestamps is respected, rejecting the block otherwise. In order for this check to take place, timestamps computed by the DCN are accompanied by threshold signatures, cryptographic gadgets used to prove that each timestamp was agreed upon by at least one honest node. Nodes in the DCN must not only be trustworthy, but also have good network conditions and be able to handle a large volume of service requests. To ensure the precision and consistency of the nodes' clocks, as well as nodes' high availability, we implement the DCN as a permissioned system, where the identity and public keys of the nodes are known to the validators. Nodes are not intended to change frequently and, by keeping the set of nodes in the system fixed, we can ensure that the nodes are reliable and that the timestamping service is accurate. 

\paragraph*{Network Model and Assumptions}
The DCN consists of $n$ nodes in a fully-connected network, such that any two nodes in the network can communicate through authenticated channels. Nodes can moreover receive external inputs, e.g., transactions from users. Each node comes equipped with a clock. We assume that node clocks are periodically realigned with real-time, which can be achieved through the use of a common external reference, such as UTC time or GPS time.

We consider an adaptive adversary that takes control during the protocol's execution of at most $f < n / 3$ nodes, causing them to deviate arbitrarily (even maliciously) from the protocol; i.e., byzantine behavior. 

We assume an estimation $\Delta_{\dcn}$ representing an upper bound on the network delay within the DCN, i.e., messages sent between the nodes \emph{should} be delivered within $\Delta_{\dcn}$ time. Similarly, we assume an estimation $\Delta_{\ext}$ for the upper bound on the external network delay, i.e., for messages sent between users and the nodes in the DCN.
We say that the network is synchronous if the message delays are \emph{always} at most $\Delta_{\dcn}$ and $\Delta_{\ext},$ and the nodes' clocks are perfectly synchronized.
If any of these conditions fail at any point, then the network is asynchronous.  
In our work, we will assume that the network is asynchronous. However, we take into account that an asynchronous assumption is often too pessimistic to model a real-life network. Hence, we aim to offer stronger guarantees during timespans when the network is synchronous, which should be the case most of the time if our estimations $\Delta_{\dcn}$ and $\Delta_{\ext}$ are faithful.
We also take into account that real clocks may fail the perfect synchrony assumption; i.e., they may have a small skew $S,$ or their local rate may vary by a factor $\theta = 1 + o(1),$ as described in~\cite{Lenzen22}. However, we assume perfect synchronization for simplicity of presentation, and we will briefly describe how our protocols can be modified to achieve the same synchronous guarantees under the weaker clock synchronization assumptions.

\paragraph*{Cryptographic Primitives}
As mentioned previously, we employ threshold signatures. In an $(\ell, n)$-threshold signature scheme, a public key is known to the $n$ nodes and also to all users and validators. Moreover, each node $\clock$ knows a unique private key that enables the generation of a partial signature $\sigma_{\clock}(m)$ for any message $m.$ The defining property of the scheme is that $\ell$ partial signatures from distinct nodes for the same message $m$ can be combined into a single signature $\sigma(m)$ that can be verified using the public key. Formally, the scheme should satisfy robustness and non-forgeability (see the full version of \cite[{Section 2.3.2}]{threshold_sig_def} for the definitions). For our purposes, we set the threshold $\ell=f+1$ and choose the BLS
scheme~\cite{boneh2001short,shoup2000practical}. 

Furthermore, we require a secret sharing scheme. In a $(k, n)$-secret sharing scheme, a secret, such as a user transaction, is divided into $n$ so-called \emph{shares}, one known to each node, such that any $k$ nodes can reconstruct the secret, while any coalition of at most $k-1$ nodes cannot learn anything about it. Formally, the information-theoretic requirement is that any $k$ shares uniquely determine the secret, while any $k - 1$ shares must be independent of the secret. Informally, given $k - 1$ shares, every possible transaction is equally likely to result in these shares. In our work, we choose $k=f + 1$ and use the Shamir scheme~\cite{shamir_how_1979}.

\paragraph*{The Transaction Submission Protocol}
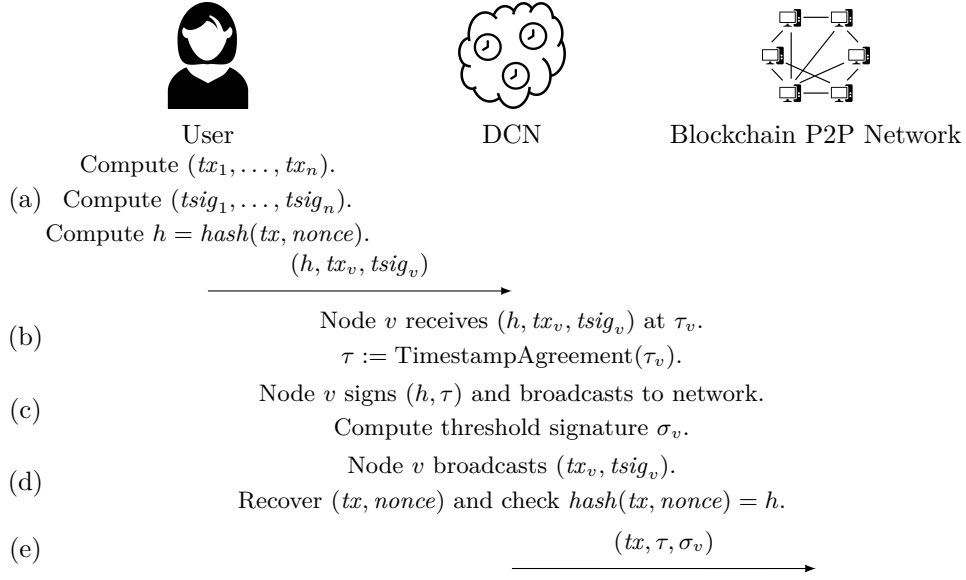
\begin{figure*}[t]
  \centering
  \begin{tikzpicture}[people/.style={minimum width=1.2cm},scale =0.8]
\node[]  at (0,-1.3) (alice) {User};
\node[] at (5,-1.3)  (vs) {DCN};
\node[] at (10,-1.3)  (eve) {Blockchain P2P Network};

\node[inner sep=0pt] (whitehead) at (0,0)  (alice1)  {\includegraphics[width=1.4cm]{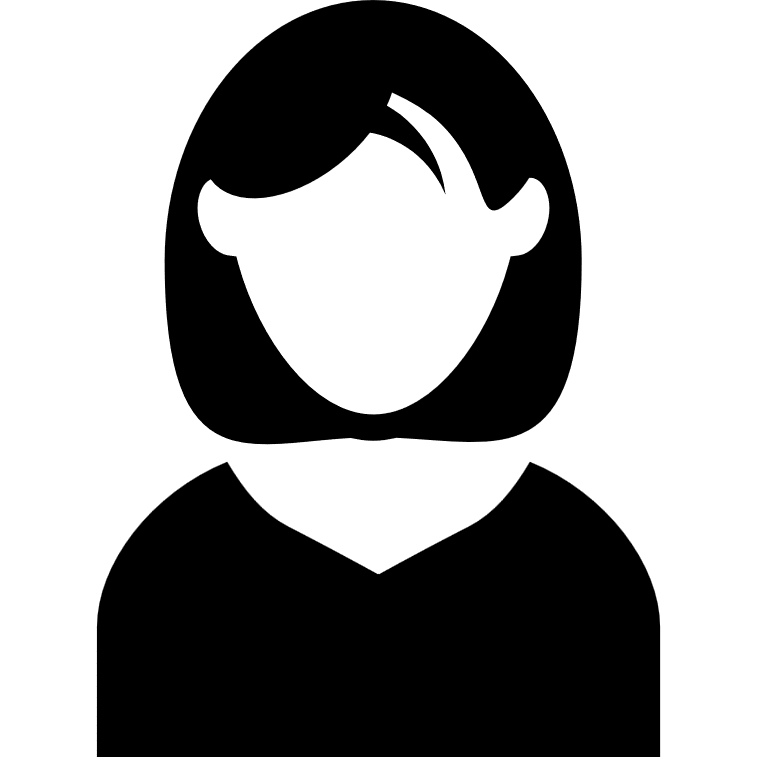}};
\node[inner sep=0pt] (whitehead) at (5,0)  (vs1)  {\includegraphics[width=1.4cm]{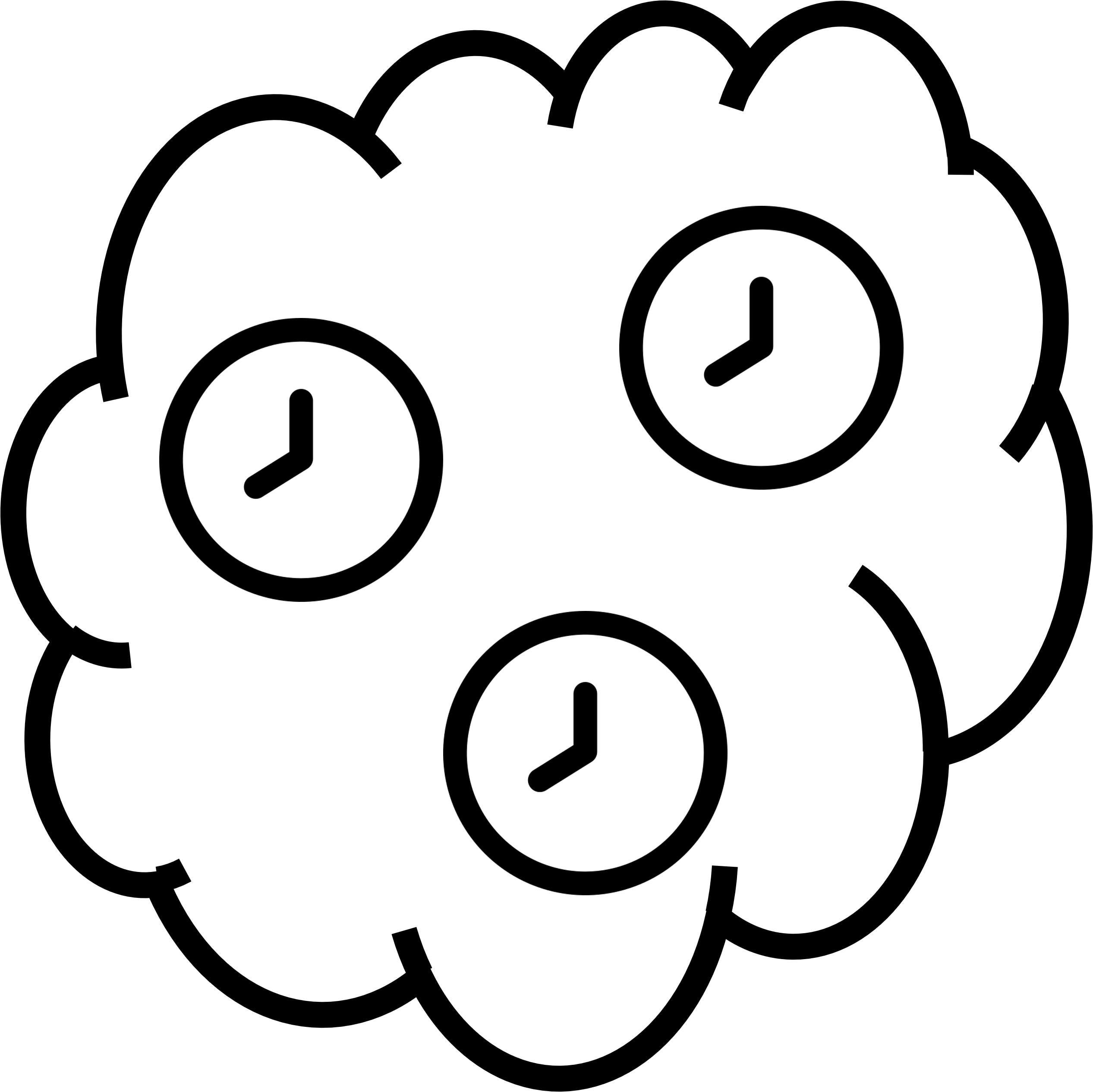}};
\node[inner sep=0pt] (whitehead) at (10,0)  (eve1)  {\includegraphics[width=1.4cm]{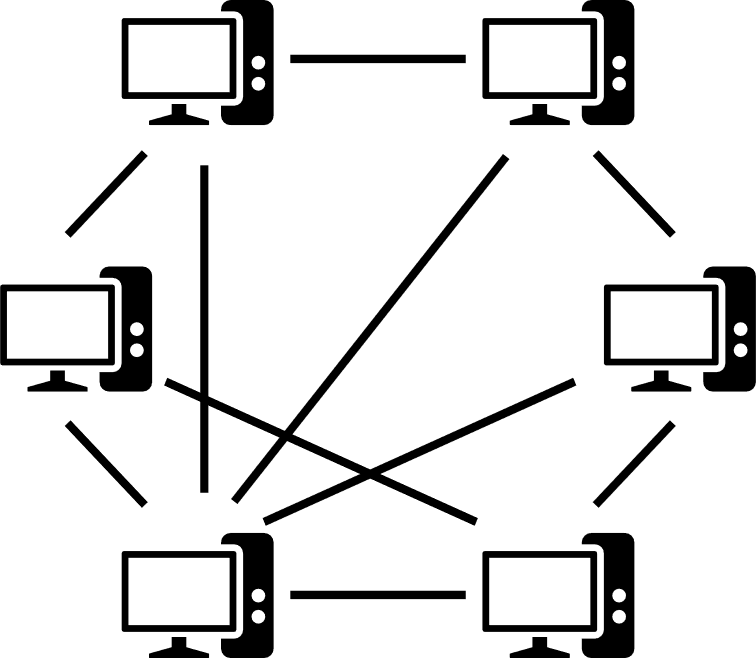}};

\node[] at (-3,-2.4) (a) {(a)};
\node[] at (0,-2.4) (a1) {\begin{tabular}{c}Compute $(\tx_1, \ldots, \tx_n).$ \\Compute $(\tsig_1, \ldots, \tsig_n).$ \\Compute $h = \hash(\tx, \nonce).$\end{tabular}};

\draw[->] ([yshift=-2.3cm]alice.south) coordinate (l1)--(l1-|vs) node[midway, above]{\begin{tabular}{c} $(h, \tx_\clock, \tsig_\clock)$\end{tabular}};

\node[] at (-3,-4.7) (b) {(b)};
\node[] at (5,-4.7) (b1) {\begin{tabular}{c} Node $\clock$ receives $(h, \tx_\clock, \tsig_\clock)$ at $\timestamp_\clock.$\\ $\timestamp := \text{\textcolor{black}{Timestamp}Agreement}(\timestamp_\clock).$ \end{tabular}};

\node[] at (-3,-5.9) (c) {(c)};
\node[] at (5,-5.9) (c1) {\begin{tabular}{c} Node $\clock$ signs $(h, \timestamp)$ and broadcasts to network.\\ Compute threshold signature $\sigma_\clock.$\end{tabular}}; 

\node[] at (-3,-7.1) (d) {(d)};
\node[] at (5,-7.1) (d1) {\begin{tabular}{c} Node $\clock$ broadcasts $(\tx_\clock, \tsig_\clock).$\\ Recover $(\tx, \nonce)$ and check $\hash(\tx, \nonce) = h.$\end{tabular}};

\node[] at (-3,-8.2) (e) {(e)};
\draw[->] ([yshift=-6.9cm]vs.south) coordinate (l1)--(l1-|eve) node[midway, above]{\begin{tabular}{c} $ (\tx, \timestamp, \sigma_v)$\end{tabular}};

\vspace{3mm}
\end{tikzpicture}
  \caption{Illustration of the transaction submission (i.e., main) protocol.}
  \label{fig:flow}\vspace{-8pt}

\end{figure*}

In this section, we formally present the protocol used when users submit transactions to the blockchain (cf.~Figure~\ref{fig:flow}), which we also refer to as the \emph{main} protocol. Assume a user wants to submit transaction $\tx,$ then the following steps are to be followed: 

\begin{enumerate}[label=(\alph*)]
    \item \label{step_1} The user generates a random nonce $\nonce.$ Then, the user splits the pair $(\tx, \nonce)$ into $n$ shares $(\tx_1, \ldots, \tx_n)$ using the $(f+1, n)$-secret sharing scheme and signs the shares with their private key to get $(\tsig_1, \ldots, \tsig_n).$ Subsequently, the user hashes the transaction together with the nonce as $h = \hash(\tx, \nonce).$ Finally, they send to each node $\clock$ the tuple $(h, \tx_\clock, \tsig_\clock).$
    \item \label{step_2} Each node $\clock$ receives $(h, \tx_\clock, \tsig_\clock)$ at some time $\timestamp_\clock.$ Together, the nodes run a \emph{Timestamp Agreement} protocol to agree on a common timestamp $\timestamp$ for transaction $\tx.$ The agreement protocol is described in detail in Section \ref{sec:agreement}.
    \item \label{step_3} Upon reaching agreement, each node signs $(h, \timestamp)$ and broadcasts the signature to the other nodes. Each node $\clock$ receives the signatures of $(h, \timestamp),$ verifies them, and uses the at least $f + 1$ valid ones to compute a threshold signature $\sigma_\clock$ for the pair $(h, \tau).$ 
    \item \label{step_4} Afterwards, each node $\clock$ broadcasts their signed share $(\tx_\clock, \tsig_\clock)$ to all other nodes. Each node receives the signed shares, verifies the signatures, and uses the at least $f + 1$ valid shares to recover the pair $(\tx, \nonce).$ Finally the node checks whether $\hash(\tx, \nonce) = h,$ aborting the protocol otherwise.
    \item \label{step_5} Each node $\clock$ now knows $\tx$ and submits it timestamped to the blockchain's peer-to-peer (P2P) network as the tuple $(\tx, \timestamp, \sigma_v).$
\end{enumerate}

The blockchain now operates with tuples of the form $(\tx, \timestamp, \sigma)$ instead of just transactions $\tx.$ For each transaction in a block, validators check the threshold signature $\sigma$ using the public keys of the nodes. Moreover, they also check that transactions are ordered in non-decreasing order by $\timestamp$ inside the block and that the lowest timestamp in the block is no lower than the highest timestamp in the previous block. 

We now provide additional intuition for the submission protocol and reasoning behind some of the design considerations. Step \ref{step_1} describes the user-sided part, while steps \ref{step_2}--\ref{step_5} describe the DCN-sided part.

In step~\ref{step_1} the user hashes $\tx$ together with a random nonce and sends it to the nodes. The nonce is required to prevent malicious actors from inferring information about $\tx$ based on past transaction data; e.g.~if a user submits similar transactions periodically, they can be identified by their hash and front-run, e.g., buying ETH every time they receive their paycheck.  
Moreover, the transaction-nonce pair is split into $n$ shares which are distributed to the $n$ nodes. This allows the DCN to recover the pair $(\tx, \nonce)$ after agreeing on timestamp $\timestamp,$ check its integrity against the hash $h,$ and submit $\tx$ to the blockchain on the user's behalf, preventing users from submitting many timestamping requests without submitting matching transactions to the blockchain, which would be the source of attacks.

In step~\ref{step_2} the DCN agrees on a common timestamp $\timestamp$ for transaction $\tx$ using the agreement protocol described later on in Section~\ref{sec:agreement}, which is efficient, robust to at most $f < n/3$ byzantine failures in both synchronous and asynchronous settings, and achieves good fairness guarantees, whose definitions we postpone to the next section.

In step~\ref{step_3} the DCN computes threshold signatures for the pair $(h, \timestamp)$ consisting of the transaction hash together with timestamp $\timestamp.$ Any valid such signature can be used to prove that at least $f + 1$ nodes have agreed on it; i.e., at least one honest node.

In step~\ref{step_4} the nodes circulate their shares to recover the pair $(\tx, \nonce).$ Note that this has to be done after agreeing on the timestamp because otherwise a byzantine node could front-run $\tx$ by submitting its own transaction and having agreement happen for it faster than for $\tx.$ Moreover, checking the hash of the pair against $h$ is required to prevent dishonest users from sending contradicting shares. Note that steps \ref{step_3} and \ref{step_4} can be implemented concurrently, but we chose not to do so for simplicity of exposition.

Finally, in step \ref{step_5} each node $\clock$ submits $\tx$ together with timestamp $\timestamp$ and threshold signature $\sigma_v$ to the blockchain, which will handle checking the signatures and ensuring that transactions are ordered by timestamp inside each block and across blocks. Note that different nodes might compute different threshold signatures $\sigma_v,$ even in the presence of no byzantine nodes, because of the choice of which individual signatures to include in $\sigma_v,$ but any valid such signature is enough to certify the tuple $(\tx, \timestamp, \sigma_v).$ We further note that validators will of course check that the transaction $\tx$ is only executed once.

We state our protocol's guarantees in the theorem below, which we prove in Section \ref{section:main-protocol-proofs}. We provide a formal definition for the term \emph{fair timestamp} in Section \ref{section:fairness-stuff}.
\begin{theorem}\label{theorem:main-protocol}
    The transaction submission protocol achieves the following properties:
    \begin{itemize}
        \item (Honest-User Liveness) If a transaction is sent by an honest user, it gets processed and submitted to the mempool eventually. Moreover, if the honest user's messages reach the nodes within $\Delta_{\ext}$ time and the synchrony assumptions hold inside the DCN for an additional $\Delta_{\dcn}$ time, the transaction gets submitted within expected $\mathcal{O}(\log \Delta_{\ext})$ communication rounds.
        \item (Integrity) If a transaction gets submitted to the mempool, the process was initiated by some user.
        \item (Unique Timestamp) If a transaction gets submitted to the mempool by two nodes, with timestamps $\timestamp$ and $\timestamp',$ then $\timestamp = \timestamp'.$
        \item (Fair Timestamp) If a transaction gets submitted to the mempool with timestamp $\timestamp,$ then $\timestamp$ is a \emph{fair} timestamp.
    \end{itemize}
\end{theorem}

\section{Timestamp Agreement and The Fairness Guarantees}\label{section:fairness-stuff} 
Nodes in the DCN need to agree on a timestamp for each transaction. This problem reduces to achieving asynchronous Byzantine Agreement (aBA), with a special Validity condition, which will allow us to argue why transactions are ordered in a fair manner.
We recall the classical definition of $\byzantineagreement,$ which requires the following properties:
\emph{(Weak~Validity)} If all honest nodes have input $\timestamp,$ no honest node outputs $\timestamp' \neq \timestamp$;
\emph{(Agreement)} If two honest nodes output $\timestamp$ and $\timestamp',$ then $\timestamp = \timestamp'$; \emph{(Termination)} Every honest node outputs with probability $1.$

While $\byzantineagreement$ is an essential building block in distributed computing, it comes with many limitations.
We first note the seminal result of \cite{FLP}, which proves that fault-tolerant $\byzantineagreement,$ even with binary inputs, cannot be solved deterministically. There is still hope however, as the distributed computing literature offers plenty of randomized $\byzantineagreement$ protocols \cite{mostefaoui2015signature,rabin1993randomized,berman1993randomized,friedman2005simple,bracha1987asynchronous,toueg1984randomized,canetti93fast,cachin2000random}.

Unfortunately, there is another limitation that prevents us from directly applying existing $\byzantineagreement$ protocols to our setting, standing in its Weak Validity condition:~this only ensures that honest nodes agree on an honest input if they joined $\byzantineagreement$ with the same input.
This pre-agreement condition is a very strong requirement in our setting, and hence nodes may often end up agreeing on timestamps proposed by corrupted nodes. Such timestamps may be too low or too high, preventing us from ensuring any kind of fair ordering.
We add that achieving a stronger condition that requires the honest nodes to always agree on some honest node's input is impossible, as one cannot distinguish between an honest node and a byzantine node following the protocol correctly, but with a corrupted input. 

\vspace{0.15cm}
\noindent\textbf{Meaningful Timestamps.}
Fortunately, there are still a few Validity variations we can consider. 
In the following definitions, we will make use of the timestamps that the nodes record when receiving messages from the user. We need to consider that, if the user is dishonest, some honest nodes might not hold such a timestamp.
Note that there is at least one honest node who has received a message from this user (otherwise the user is essentially not sending a transaction). Then, let $\tau_{\max}$ denote the latest point in time recorded by an honest node when receiving this user's message. In the definitions, we assume that, if an honest node does not receive such a message, its input is $\tau_{\max}.$ We stress that this assumption is strictly for simplicity of presentation and is not used in our protocols or their analysis.

With this convention in mind, we may provide stronger Validity definitions.
In our setting, ensuring that honest nodes' outputs are in their inputs' range is already meaningful (\emph{Honest-Range~Validity}). This enables the order fairness definition below, discussed in \cite{zhang2020byzantine}. 
\begin{definition}[Honest-Range Fairness]
    Let $\tx$ and $\tx'$ denote two transactions.
    If all honest nodes receive the hash of $\tx$ before any honest node receives the hash of $\tx',$  then $\tx$ will be ordered before $\tx'.$
\end{definition}
Honest-Range Validity has been studied in the synchronous setting \cite{vaidya2013byzantine}. In the asynchronous setting, however, this condition has only been considered under much weaker Agreement requirements, which allow the honest outputs to be $\varepsilon$-close for some predefined $\varepsilon > 0$; see \cite{OPODIS:AAD04}.

One could hope that a stronger order-fairness definition is possible. Our first attempt is as follows:~if, at some time $\tau,$ most honest nodes have received the hash of some transaction $\tx,$ while most honest nodes are yet to receive the hash of some transaction $\tx',$ then $\tx$ should be ordered before $\tx'.$ We express this condition with the help of the medians of the honest nodes' receipt timestamps: 
\begin{definition}[Median Fairness]\label{def:attempt-fair}
    Suppose the hashes of transactions $\tx$ and $\tx'$ are received by the honest nodes at times $\timestamp_1 \leq \timestamp_2 \leq \ldots \leq \timestamp_{n-f}$ and resp. $\timestamp'_1 \leq \timestamp'_2 \leq \ldots \leq \timestamp'_{n-f}.$ 
    Then, if $\tau_\mu < \tau_{\mu}',$ where $\mu = \lceil (n - f) / 2 \rceil$ denotes the index of the median, $\tx$ will be ordered before $\tx'.$
\end{definition}

To achieve this order fairness definition, we need honest nodes to agree on the median of their timestamps. Consider the ($\delta$-Median Validity) condition below, introduced by Stolz and Wattenhofer in \cite{stolz2016byzantine}, for $n > 3f.$
\begin{itemize}
\item{($\delta$-Median Validity) Assume the honest inputs are arranged in non-decreasing order in an array $T,$ and $T_i$ is the $i$-th value in $T.$ If
an honest node outputs $\timestamp,$ then $\timestamp \in [T_{\mu - \delta}, T_{\mu + \delta}]$ (i.e., $\timestamp$ is $\delta$-positions-close to $T_{\mu}$), where $\mu = \lceil (n - f) / 2 \rceil.$}
\end{itemize}

Then, Median Fairness requires $0$-Median Validity. This definition however cannot be achieved even in a synchronous network, as stated in Lemma~\ref{lemma:sync-delta}, following directly from \cite{stolz2016byzantine,melnyk2018byzantine}.
\begin{lemma} \label{lemma:sync-delta}
    If $n > 3f$ and $\delta < \lceil f/2 \rceil,$ there is no synchronous protocol achieving Termination and $\delta$-Median Validity.
\end{lemma}

We therefore weaken our Median Fairness definition to allow some error.
\begin{definition}[$\delta$-Median Fairness]
    Suppose the hashes of transactions $\tx$ and $\tx'$ are received by the honest nodes at times $\timestamp_1 \leq \timestamp_2 \leq \ldots \leq \timestamp_{n-f}$ and $\timestamp'_1 \leq \timestamp'_2 \leq \ldots \leq \timestamp'_{n-f}$ respectively.
    Let $\mu = \lceil (n - f) / 2 \rceil$ denote the index of the median.
    Then, if $\tau_{\mu + \delta} < \tau'_{\mu - \delta},$ transaction $\tx$ will be ordered before transaction $\tx'.$
\end{definition}

We note that $\delta$-Median Validity has only been considered in the synchronous model \cite{stolz2016byzantine,melnyk2018byzantine}, meaning that even the slightest increased network delay may cause the protocols of \cite{stolz2016byzantine,melnyk2018byzantine}, which achieve $\delta$-Median Validity, to completely fall apart. This motivates us to study $\delta$-Median Validity in the asynchronous model. First, we show a lower bound on the $\delta$ achievable for the asynchronous case.
Later on, we will further show this bound to be tight.

\begin{restatable}{lemma}{asyncDelta}\label{lemma:async-delta}
If $n > 3f$ and $\delta < f,$ there is no asynchronous protocol achieving Termination and $\delta$-Median Validity.
\end{restatable}

\begin{proof}

    We assume that there is a protocol $\Pi$ achieving $\delta$-Median Validity and Termination.
    Let $\mu = \lceil (n - f) / 2 \rceil,$ and let $v$ denote an honest node. The input value of node $v$ will be  $2f + 1.$ We define the following scenarios:
    \begin{enumerate}[label=(\alph*)]
        \item \label{honest-middle} The $n - f$ honest nodes have inputs $f + 1, f + 2, \dots n,$ and the corrupted parties do not participate in the protocol. Then, $v$ must output a value in $[f + \mu - \delta, f + \mu + \delta].$
        \item \label{honest-left} The $n - f$ honest nodes have inputs $1, 2, \dots, n - f \geq 2f + 1.$ The $f$ corrupted nodes follow the protocol correctly with inputs  $n - f + 1, n - f + 2, \dots, n,$ while the messages of the honest nodes holding the $f$ lowest inputs are delayed. Here, $v$ should output a value in $[\mu - \delta, \mu + \delta].$ However, since from node $\clock$'s perspective, this scenario is indistinguishable from Scenario \ref{honest-middle}, $v$ must output a value in $[\mu - \delta, \mu + \delta] \cap [f + \mu - \delta, f + \mu + \delta] = [f + \mu - \delta, \mu + \delta].$
        \item The $n - f$ honest nodes have inputs $2 f + 1, 2 f + 2, \dots, n + f.$ The $f$ corrupted nodes follow the protocol correctly with inputs $f + 1, f + 2, \dots, f,$ while the messages of the $f$ honest nodes holding the $f$ highest inputs are delayed. Here, $v$ should output a value in $[2f + \mu - \delta, 2f + \mu + \delta].$ Note that, for node $\clock,$ this scenario is indistinguishable from Scenario \ref{honest-middle} and Scenario \ref{honest-left}. Therefore, node $\clock$ must output a value in $[f + \mu - \delta, \mu + \delta] \cap [2f + \mu - \delta, 2f + \mu + \delta] = [2f + \mu - \delta, \mu + \delta].$
    \end{enumerate}

    Since $\delta < f,$ we obtain that $\mu + \delta < \mu + f < 2f + \mu - \delta,$ therefore the interval $[2f + \mu - \delta, \mu + \delta]$ containing node $\clock$'s output is empty. This contradicts that $\Pi$ achieves Termination.
\end{proof}

Lemma~\ref{lemma:async-delta} showcases an important limitation, namely, in a purely asynchronous network, if $n = 3f + 1,$ one can only hope to achieve Honest-Range Validity, as in this case $f$-Median Validity degenerates to Honest-Range Validity.
We note here that previous work has shown that a single byzantine node can manipulate the median~\cite{kelkar2020order}. However, as timestamps satisfying $f$-Median Validity are still in the honest range and as transactions are not visible during ordering, we do not see it as a threat.

\vspace{0.15cm}
\noindent \textbf{Defining Timestamp Agreement.}
\textcolor{black}{To mitigate the limitation posed by Lemma \ref{lemma:async-delta}, we take into account that real-world networks are not as unreliable as the asynchronous model. Hence, we aim to provide better guarantees if the network \emph{happens to be synchronous}.} We investigate whether we can achieve best-of-both-worlds guarantees, in line with many recent works \cite{TCC:BluKatLos19,Crypto:BLL20,PODC:ApChCh22,PODC:GhLiWa22,TCC:DHL21}. That is, we investigate whether there is an asynchronous protocol ensuring $f$-Median Validity that can additionally offer the stronger guarantee of $\lceil f/2 \rceil$-Median Validity if the network \emph{happens to be synchronous} for sufficient time. Therefore, we introduce the following variant of $\byzantineagreement.$
\begin{definition}[Timestamp Agreement]
    An $n$-nodes protocol, where each node may hold an integer timestamp as input, achieves Timestamp Agreement ($
    \timestampagreement$) if, even when $f$ of the nodes are corrupted, it achieves Agreement, $f$-Median Validity, and the following hold:
    \begin{itemize}
        \item if all honest nodes hold inputs, then all honest node obtain outputs with probability $1;$
        \item if less than $f + 1$ honest nodes hold inputs, then no honest node obtains output;
        \item if the synchrony assumptions hold for a sufficient amount of time and all honest parties receive their inputs accordingly, then $\lceil f/2 \rceil$-Median Validity is achieved.
    \end{itemize}
\end{definition}

\textcolor{black}{We note that we have proposed this definition taking into account that the user is not necessarily honest, and hence may not provide all honest nodes with inputs. If this is the case, our protocol still maintains $f$-Median Validity and Agreement. For the timestamp submission protocol, this implies that, if a dishonest user's transaction gets submitted to the chain, then the unique timestamp assigned to it still fits our $f$-Median Fairness definition. Hence, such adversarial behavior does not bring the dishonest user any real advantage.}

\textcolor{black}{We may now also define the term \emph{fair} timestamp, used in Theorem \ref{theorem:main-protocol}:~it is a timestamp satisfying $f$-Median Validity, and, if synchrony assumptions hold, $\lceil f/2 \rceil$-Median Validity.}

\section{The Timestamp Agreement Protocol} \label{sec:agreement}

In this section, we present our protocol achieving Timestamp Agreement secure against $f < n /3$ byzantine corruptions. 
Formally, we obtain the result below. 
Recall once again that there is no deterministic protocol achieving asynchronous Timestamp Agreement, a fact following directly from FLP \cite{FLP}.
\begin{theorem} \label{thm:timestamp-agreement}
    There is an $n$-nodes randomized protocol $\Pi_{\timestampagreement}$ achieving $\timestampagreement$ resilient against $f < n / 3$ byzantine corruptions. $\Pi_{\timestampagreement}$ has expected round complexity $\mathcal{O}(\log(\timestamp_{\max} - \timestamp_{\min})),$ where $\timestamp_{\min}$  and $\timestamp_{\max}$ denote the lowest and the highest honest inputs respectively. To achieve $\lceil f / 2 \rceil$-Median Validity, the synchrony assumptions must to hold for $\Delta_{\ext} + \Delta_{\dcn}$ time.
\end{theorem}

\textcolor{black}{
Our protocol $\Pi_{\timestampagreement}$ consists of three steps. First, each node obtains a value satisfying $\delta$-Median Validity. This is the only step where synchrony assumptions are required to achieve $\delta = \lceil f/2 \rceil$ instead of $\delta = f.$ In the second step, nodes obtain very close values (they agree up to an error of $\varepsilon < 0.5$) within the range of values that honest nodes obtained in the first step. Agreement is then achieved in the last step, where each node decides whether to round its value obtained in the second step up or down. This will be done using $\byzantineagreement$ on the rounding option's parity. In the following, we describe each step of $\Pi_{\timestampagreement}$ in detail. }
\newline

\noindent\textbf{Step 1: $\delta$-Median Validity.}
We first design a protocol $\Pi_{\init}$ that only focuses on achieving $\delta$-Median Validity (while Agreement is covered by the subsequent steps).

Concretely, nodes send their input value to every party.
\textcolor{black}{
To obtain a good estimation on the honest inputs' median, the nodes aim to receive as many honest inputs as possible.
If the network is asynchronous, one may only expect to receive $n - f$ values.
On the other hand, if the network is synchronous, and the user initiated the transaction at some time $\tau,$ all honest inputs are received by time $\tau + \Delta_{\ext} + \Delta_{\dcn}.$ 
Then, nodes wait until they have received timestamps from at least $n - f$ nodes, and, until at least $\Delta_{\ext} + \Delta_{\dcn}$ time has passed since they have received the user's message. This way, if the synchrony assumptions hold, every honest timestamp is received.
}

Hence, each node $\clock$ collects $n - f + k$ timestamps, where $0 \leq k \leq f,$ and arranges them in an array $R$ in non-decreasing order. 
\textcolor{black}{
If the network is synchronous, at most $k$ of these values are corrupted. These may be lower than any honest input, hence shifting the honest median with at most $k$ positions to the right, or higher than any honest input. Therefore, the honest median is in the subarray $R_{\mu}, R_{\mu + 1}, \ldots, R_{
\mu + k},$ where $R_i$ denotes the $i$-th lowest value in $R$ and $\mu = \lceil (n - f) / 2 \rceil.$ Then, to obtain a value that is $\lceil f/2 \rceil$-positions-close to the median, $v$ outputs $\timestamp_{\mu} := R_{\mu + \lfloor k / 2 \rfloor},$ i.e., the median of the subarray $R_{\mu}, R_{\mu + 1}, \ldots, R_{
\mu + k}.$
}

\textcolor{black}{
If the synchrony assumptions fail, however, the $n - f + k$ values from $R$ might come from $f$ corrupted nodes, and $n - 2f + k$ honest nodes. The $f - k$ missing honest timestamps provide the corrupted nodes with more power:~shifting the honest median $f$ positions to the right, or $f - k$ positions to the left. Regardless, the chosen output $\timestamp_{\mu} := R_{\mu + \lfloor k / 2 \rfloor}$ still ensures that $f$-Median Validity holds.
}

We formally present the code of $\Pi_{\init}$ below. 
We note that this is the only step requiring synchrony for achieving $\lceil f/2 \rceil$-Median Validity. In order to achieve the same guarantee even if the nodes' clocks are not perfectly synchronized, we may replace the waiting time by $\theta \cdot (\Delta_{\ext} + \Delta_{\dcn}),$ to ensure that the fastest node waits long enough.

\begin{protocolbox}{$\Pi_{\init}$}
	\algoHead{Code for node $\clock$ with input timestamp $\timestamp_{\inputt}$}
	\begin{algorithmic}[1]
            \State Send your input $\timestamp_\inputt$ to all nodes.
            \State After at least $\Delta_{\ext} + \Delta_{\dcn}$ time, and when $n - f + k$ timestamps ($0 \leq k \leq f$) are received:
            \State \hspace{0.5cm} $R$ := an array containing the timestamps received, ordered non-decreasingly.
            \State \hspace{0.5cm} Output $\timestamp_{\mu} := R_{\lceil (n - f) / 2 \rceil + \lfloor k / 2 \rfloor}.$
	\end{algorithmic}
\end{protocolbox}

We may now state and prove the properties of $\Pi_{\init}.$

The next property enables us to ensure safety guarantees even when the user initiating the process is dishonest. It follows immediately from line $2$ of $\Pi_{\init}.$
\begin{lemma} \label{lemma:no-outputs}
    \textcolor{black}{
        If less than $f + 1$ honest nodes provide inputs $\timestamp_{\inputt},$ then no honest node outputs. Otherwise, if each honest node provides an input $\timestamp_{\inputt},$ then all honest nodes output $\timestamp_{\mu}.$
    }
\end{lemma}

The following lemmas show that the nodes indeed obtain values satisfying the desired Validity guarantees.

\begin{restatable}{lemma}{asyncValidityAA}\label{lemma:async-validity}
\textcolor{black}{If an honest node outputs a timestamp $\timestamp_{\mu},$ then $\timestamp_{\mu}$ satisfies $f$-Median Validity.}
\end{restatable}

\begin{proof}
    If an honest node $v$ has obtained a timestamp $\timestamp_{\mu},$ then it has received $n - f + k$ values, where $0 \leq k \leq f.$ Out of these, at least $n - 2f + k$ values are honest.
    
    Let $T$ denote the array of honest inputs arranged in non-decreasing order.
    We show that $T_{\lceil (n - f) / 2 \rceil - f} \leq \timestamp_{\mu} = R_{\lceil (n - f) / 2 \rceil + \lfloor k/2 \rfloor} \leq T_{\lceil (n - f) / 2 \rceil + f}.$
    
    For the upper bound, note that $R$ may miss $f - k$ out of the values in $T,$ hence at most $f - k$ of the values $T_i$ with $i \leq \lceil (n - f) / 2 \rceil + \lfloor k/2 \rfloor.$ This implies that  $R_{\lceil (n - f) / 2 \rceil + \lfloor k/2 \rfloor} \leq T_{\lceil (n - f) / 2 \rceil + \lfloor k/2 \rfloor + (f - k)} \leq T_{\lceil (n - f) / 2 \rceil + f}.$ 
    
    For the lower bound, note that $R$ contains at most $f$ corrupted values, hence at most $f$ additional values that are lower than $T_{\lceil (n - f) / 2 \rceil + \lfloor k/2 \rfloor}.$ Then, we obtain that $R_{\lceil (n - f) / 2 \rceil + \lfloor k/2 \rfloor} \geq T_{\lceil (n - f) / 2 \rceil + \lfloor k/2 \rfloor - f} \geq T_{\lceil (n - f) / 2 \rceil - f}.$

\end{proof}

We now show that $\Pi_{\init}$ achieves $\lceil f/2 \rceil$-Median Validity if the synchrony assumptions hold, using a similar argument to the proof of Lemma \ref{lemma:sync-validity}. The key difference is that at least $n - f$ of the values received are honest (as opposed to $n - 2f + k$).

\begin{restatable}{lemma}{syncValidityAA}\label{lemma:sync-validity}
    If all honest nodes obtain inputs $\timestamp_{\inputt}$ and join $\Pi_{\init}$ between time $\timestamp_{\startt}$ and time $\timestamp_{\startt} + \Delta_{\ext},$ and the synchrony assumptions hold until time $\timestamp_{\startt} + \Delta_{\ext} + \Delta_{\dcn},$ then all honest nodes output timestamps $\timestamp_{\mu}$ satisfying $\lceil f/2 \rceil$-Median Validity. 
\end{restatable}

\begin{proof}
    We first show that all honest timestamps are received by time $\timestamp_{\startt} + \Delta_{\ext} + \Delta_{\dcn}.$ Each honest node sends its input to all other nodes by time $\timestamp_{\startt} + \Delta_{\ext}.$
    Since the network is synchronous, these values are received within $\Delta$ time, hence by time $\timestamp_{\startt} + \Delta_{\ext} + \Delta_{\dcn}.$
    Then, since all honest nodes start the execution of the protocol at time at least $\timestamp_{\startt},$ the protocol ensures that each honest node waits until time at least $\timestamp_{\startt} + \Delta_{\ext} + \Delta_{\dcn},$ and hence receives all honest timestamps.

    Then, for every honest node, $R$ contains all honest values, and $0 \leq k \leq f$ values from corrupted nodes. If $T$ denotes the array of honest timestamps arranged in non-decreasing order, we need to show that $T_{\lceil (n - f) / 2 \rceil - \lceil f/2 \rceil} \leq R_{\lceil (n - f) / 2 \rceil + \lfloor k/2 \rfloor} \leq T_{\lceil (n - f) / 2 \rceil + \lfloor f/2 \rfloor}.$

    We first focus on the upper bound: since $R$ contains all values $T_i,$ and $k \leq f,$ the inequality $R_{\lceil (n - f) / 2 \rceil + \lfloor k/2 \rfloor} \leq T_{\lceil (n - f) / 2 \rceil + \lfloor k/2 \rfloor} \leq T_{\lceil (n - f) / 2 \rceil + \lfloor f/2 \rfloor}$ holds.

    For the lower bound, we note that $R$ contains at most $k + \mu +  \lfloor k/2 \rfloor$ values lower than $T_{\lceil (n - f) / 2 \rceil + \lfloor k/2 \rfloor}.$ Out of these  $k + \mu +  \lfloor k/2 \rfloor$ values, at most $k$ are corrupted.
    This means that $R_{\lceil (n - f) / 2 \rceil + \lfloor k/2 \rfloor} \geq T_{\lceil (n - f) / 2 \rceil + \lfloor k/2 \rfloor - k} = T_{\lceil (n - f) / 2 \rceil - \lceil k/2 \rceil} \geq T_{\lceil (n - f) / 2 \rceil - \lceil f/2 \rceil},$ which concludes our proof.
\end{proof}

\noindent\textbf{Step 2: Agreement up to a small error.}
Honest nodes have obtained timestamps $\timestamp_{\mu}$ satisfying $\delta$-Median Validity via $\Pi_{\init}.$ We now take a step towards achieving Agreement. We make use of an asynchronous protocol $\Pi_{\approxagreement}$ achieving Approximate Agreement \cite{OPODIS:AAD04}.
That is, $\Pi_{\approxagreement}$ ensures that, for any given $\varepsilon > 0,$ honest nodes obtain $\varepsilon$-close values $\timestamp_\approxagreement$ within the range of their values $\timestamp_{\mu}$ (maintaining $\delta$-Median Validity). Lemma \ref{lemma:AA-median} states the properties of $\Pi_{\approxagreement},$ and follows directly from \cite{OPODIS:AAD04}. In our case, any constant $\varepsilon < 0.5$ suffices.
\begin{lemma}\label{lemma:AA-median}
    If an honest node outputs $\timestamp_{\approxagreement},$ then 
 $\timestamp_{\approxagreement}$ is within the range of timestamps $\timestamp_{\mu}$ obtained by honest nodes in $\Pi_{\init}.$ If two honest nodes output $\timestamp_{\approxagreement}$ and $\timestamp_{\approxagreement}',$ then 
    $\abs{\timestamp_{\approxagreement} - \timestamp_{\approxagreement}'} < \varepsilon < 0.5.$ \textcolor{black}{In addition, if less than $f + 1$ honest nodes hold timestamps $\timestamp_{\mu},$ then no honest node outputs; while if all honest nodes hold timestamps $\timestamp_{\mu},$ then all honest nodes output}.
\end{lemma}

The guarantee on obtaining outputs only when at least $f + 1$ honest nodes participate follows from the fact that $\Pi_{\approxagreement}$ requires nodes to wait for messages from $n - f$ distinct nodes. In addition, properties on honest nodes' outputs (if any) are ensured even if not all honest nodes participate since $\Pi_{\approxagreement}$ is an asynchronous protocol. Concretely, this is because this setting is indistinguishable from a scenario where the non-participating honest nodes are simply delayed.

 We add that $\Pi_{\approxagreement}$ does not make any assumption on the range of honest values $\timestamp_{\mu}$ to achieve these guarantees. It runs in iterations allowing the honest values to converge. If all honest nodes hold inputs $\timestamp_{\mu}$ and range size of these inputs is $\Delta_{\mu}$ ($\leq $ the difference between the times when the transaction hash is delivered to the honest nodes), then $\Pi_{\approxagreement}$ runs for $\mathcal{O}(\log (\Delta_{\mu} / \varepsilon))$ iterations.
 Each iteration consists of a constant number of communication rounds, and incurs message complexity $\mathcal{O}(n^3).$ Therefore, the round complexity of $\Pi_{\approxagreement}$ is  $\mathcal{O}(\log \Delta_{\mu}),$ and the message complexity is $\mathcal{O}(n^3 \log \Delta_{\mu}).$
\newline

\noindent \textbf{Step 3: Rounding.} Honest nodes have obtained $\varepsilon$-close values $\timestamp_{\approxagreement}$ satisfying $\delta$-Median Validity.
As depicted in Figure \ref{figure:even-odd-cases}, 
since $\varepsilon < 0.5,$ the range of honest values $\timestamp_{\approxagreement}$ either:
\begin{enumerate}[label=(\alph*)]
    \item\label{case:even} contains an even integer $\alpha$ such that $\abs{\alpha - \timestamp_{\approxagreement}} < 0.5$ for all honest values $\timestamp_{\approxagreement}$;
    \item\label{case:odd} contains an odd integer $\alpha$ such that $\abs{\alpha - \timestamp_{\approxagreement}} < 0.5$ for all honest values  $\timestamp_{\approxagreement}$;
    \item\label{case:between} is between two integers: $\alpha \leq \timestamp_{\approxagreement} \leq \alpha + 1$ for all honest values $\timestamp_{\approxagreement}.$
\end{enumerate}

\begin{figure}[h]\vspace{-15pt}
\centering
\includegraphics[width=0.38\textwidth]{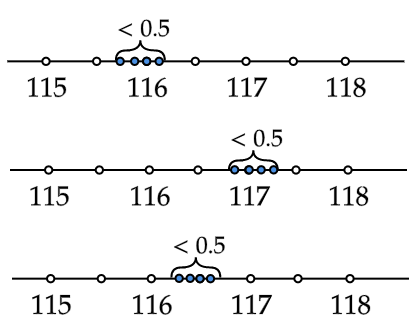} 
\caption{A few examples of possible outputs obtained in $\Pi_{\approxagreement}.$ In the top and middle examples, representing cases 
\ref{case:even} and \ref{case:odd} respectively, honest outputs are close to a single integer. In the bottom example, representing case \ref{case:between}, some honest outputs are closer to $116,$ while some are closer to $117.$}\label{figure:even-odd-cases} 

\end{figure}

Then, the problem of achieving Agreement comes down to enabling the honest nodes to choose between rounding down or rounding up their values $\timestamp_{\approxagreement}.$ Making this decision for cases \ref{case:even} and \ref{case:odd} is trivial: honest nodes simply round their value $\timestamp_{\approxagreement}$ to the closest integer. Case \ref{case:between}, however, requires solving $\byzantineagreement.$
We therefore employ the randomized protocol $\Pi_{\byzantineagreement}$ of \cite{mostefaoui2015signature} that achieves $\byzantineagreement$ with binary inputs in expected round complexity $\mathcal{O}(1),$ with message complexity $\mathcal{O}(n^2).$ Note that we do not use $\byzantineagreement$ to decide on rounding either up or down, since this would break Agreement in cases \ref{case:even} and \ref{case:odd}. Instead, we use $\byzantineagreement$ to decide on the parity of the final rounding option.
\textcolor{black}{Once again, we note that if the user is dishonest and not all honest nodes were able to reach this stage, protocol $\Pi_{\byzantineagreement}$ still offers guarantees. Namely, if less than $f + 1$ honest nodes have reached this stage, then no honest node obtains an output (as honest nodes are forced to wait for messages from $n - f$ distinct nodes). Otherwise, if honest nodes obtain outputs, these outputs still satisfy Weak Validity and Termination. This is the case even if not all honest nodes participate, since such a setting is indistinguishable from a scenario where the non-participating honest nodes' messages are simply delayed, and the guarantees of $\Pi_{\byzantineagreement}$ hold under asynchrony.}

Each node $\clock$ that has obtained a timestamp $\timestamp_{\approxagreement}$ picks two integers $\alpha$ and $\alpha + 1$ such that $\alpha \leq \timestamp_{\approxagreement} < \alpha + 1.$ Out of these two values, $\clock$ picks the one that is closer to its $\timestamp_{\approxagreement}$ as an initial rounding option, denoted by $\beta.$ Then, $\clock$ joins $\Pi_{\byzantineagreement}$ with input $b,$ representing the parity of $\beta,$ and may obtain output $b'.$ If $b' = b,$ it outputs $\beta,$ and otherwise it outputs its second rounding option.
In cases \ref{case:even} and \ref{case:odd}, honest nodes pick the same value $\beta.$ They join $\Pi_{\byzantineagreement}$ with the same input bit $b,$ and Weak Validity ensures they output $b' = b.$ Therefore, if sufficiently many honest nodes reached this stage, all participating honest nodes output $\beta,$ which still satisfies $\delta$-Median Validity.
In case \ref{case:between}, all honest nodes that reach this stage pick the same value $\alpha.$ In this case, because the input timestamps $\timestamp_{\inputt}$ are integers, both $\alpha$ and $\alpha + 1$ satisfy $\delta$-Median Validity. Even if honest nodes make a different choice for $\beta,$ $\Pi_{\byzantineagreement}$ allows them to decide on the same bit $b',$ hence they output the same rounding option.

We may now provide the formal code of our Timestamp Agreement protocol $\Pi_{\timestampagreement}.$ We define the constant $\varepsilon = 0.49,$ but any choice of $\varepsilon < 0.5$ suffices.

\begin{protocolbox}{$\Pi_{\timestampagreement}$}
	\algoHead{Code for node $\clock$ receiving a transaction at time $\timestamp_{\inputt}$}
	\begin{algorithmic}[1]
            \State Join $\Pi_{\init}$ with input $\timestamp_{\inputt}.$ Upon obtaining $\timestamp_{\mu}$ via $\Pi_{\init}$:
            \State \hspace{0.5cm} Join $\Pi_{\approxagreement}^{\varepsilon}$ with input $\timestamp_\mu.$ Upon obtaining output $\timestamp_{\approxagreement}$ in  $\Pi_{\approxagreement}^{\varepsilon}$:
		\State \hspace{1cm} Let $\alpha$ be an integer such that $\alpha \leq \timestamp_{\approxagreement} < \alpha + 1.$
            \State \hspace{1cm} If $\timestamp_{\approxagreement} - \alpha < \alpha + 1 - \timestamp_{\approxagreement},$ set $\beta = \alpha$ and $\beta' = \alpha + 1.$
            \State \hspace{1cm} Otherwise, set $\beta = \alpha + 1$ and $\beta' = \alpha.$
            \State \hspace{1cm} Set $b = 0$ if $\beta$ is even, and $b = 1$ if $\beta$ is odd.
            \State \hspace{1cm} Join $\Pi_{\byzantineagreement}$ with input $b.$ Upon obtaining output $b'$ via $\Pi_{\byzantineagreement}$:
            \State \hspace{1.5cm} If $b = b'$ , set $\timestamp_{\outputt} = \beta.$ Otherwise, set $\timestamp_{\outputt} = \beta'.$ Output $\timestamp_{\outputt}$ and terminate.
	\end{algorithmic}
\end{protocolbox}

We now focus on proving Theorem \ref{thm:timestamp-agreement}. In Lemma \ref{lemma:outputs-timestamp-aa}, we show that $\Pi_{\timestampagreement}$ indeed achieves Timestamp Agreement.
Then, Lemma \ref{lemma:TA-complexity} focuses on the round complexity. The requirement of synchrony assumptions holding only for $\Delta_{\ext} + \Delta_{\dcn}$ in order to achieve $\lceil f / 2 \rceil$-Median Validity is given by $\Pi_{\init},$ since the subsequent steps of $\Pi_{\timestampagreement}$ are fully asynchronous.

\begin{restatable}{lemma}{outputsTimestampAA}\label{lemma:outputs-timestamp-aa}
If less than $f + 1$ honest nodes hold inputs $\timestamp_{\inputt},$ then no honest node outputs.

Otherwise, honest nodes that output have obtained the same value $\timestamp_{\outputt}$ satisfying $\delta$-Median Validity, with $\delta = 
\lceil f / 2\rceil$ if the synchrony assumptions held for $\Delta_{\ext} + \Delta_{\dcn}$ time at the beginning of the protocol's execution, and $\delta = f$ otherwise.

In addition, if all honest nodes hold inputs $\timestamp_{\inputt},$ then all honest nodes output.
\end{restatable}
\begin{proof}
    Lemma \ref{lemma:no-outputs} ensures that $f + 1$ honest nodes holding inputs $\timestamp_{\inputt}$ are necessary in order to obtain outputs. In the following, we assume this was the case.
    
    Lemma \ref{lemma:AA-median} ensures that honest nodes obtaining outputs (meaning all honest nodes if all of them had inputs $\timestamp_{\inputt}$) have obtained $\varepsilon$-close approximations $\timestamp_{\approxagreement}$ for $\varepsilon < 0.5.$ These approximations are within the range of honest values $\timestamp_{\inputt},$ and hence satisfy $\delta$-Median Validity, as ensured by Lemma \ref{lemma:async-validity} and Lemma \ref{lemma:sync-validity}.
    
    Then, we need to consider two cases: when all obtained honest approximations are between two consecutive integers, and when some honest approximations are lower than an integer, while some are higher.

    If there is some integer $\gamma$ such that $\gamma \leq \timestamp_{\approxagreement} < \gamma + 1$ for all obtained honest approximations $\timestamp_{\approxagreement},$ then  honest nodes obtain $\alpha = \gamma$ and $\alpha + 1 = \gamma + 1.$ Regardless of the chosen $\beta$ and bit $b,$ honest nodes obtain in $\Pi_{\byzantineagreement}$ the same bit $b'$ which refers to the same value: either $\gamma$ for all honest nodes that reached this stage, or $\gamma + 1$ for all honest nodes that reached this stage. Hence, these honest nodes output the same timestamp.
    It remains to show that the output timestamp is in the range of honest inputs. If all honest nodes that reached this stage have obtained $\timestamp_{\approxagreement} = \gamma,$ then they joined $\Pi_{\byzantineagreement}$ with the same input $b$ representing $\gamma$'s parity, and hence they output the same $\gamma$ in the honest range according to Lemma \ref{lemma:AA-median}.
    Otherwise, if at least one honest node has obtained $\gamma < \timestamp_{\approxagreement} < \gamma + 1,$ we take into account that the honest inputs are integers. Lemma \ref{lemma:AA-median} then implies that both $\gamma$ and $\gamma + 1$ are in the honest inputs' range.

    Otherwise, there is some integer $\gamma$ such that $\gamma \leq \timestamp_{\approxagreement} < \gamma + 1$ for some honest approximation $\timestamp_{\approxagreement}$ and $\gamma + 1 \leq \timestamp_{\approxagreement}' < \gamma + 2$ for some honest approximation $\timestamp_{\approxagreement}.$ Note that, in this case, Lemma \ref{lemma:AA-median} ensures that $\gamma + 1$ is in the range of the honest nodes' inputs.     
    In addition, since Lemma \ref{lemma:AA-median} ensures $\timestamp_{\approxagreement}' -  \timestamp_{\approxagreement} < 0.5,$ both $\gamma + 1 - \timestamp_{\approxagreement} < 0.5$ and $\timestamp_{\approxagreement}' - (\gamma + 1) < 0.5$ hold. This applies to all honest nodes that reached this stage: namely, all these honest nodes choose the same $\beta = \gamma + 1$ and therefore join $\Pi_{\byzantineagreement}$ with the same bit $b.$ Then, $\Pi_{\byzantineagreement}$ ensures all honest nodes that reached this stage output $b' = b$ and output $\gamma + 1.$

    If all honest nodes had inputs $\timestamp_{\inputt},$ all honest nodes have obtained outputs in $\Pi_{\byzantineagreement},$ and therefore all honest nodes output in $\Pi_{\timestampagreement}.$
\end{proof}

The round complexity of $\Pi_{\timestampagreement}$ follows from the fact that $\Pi_{\approxagreement}$ ensures termination within $\mathcal{O}(\log(\timestamp_{\max} - \timestamp_{\min}))$ rounds, if honest nodes' inputs are between $\timestamp_{\min}$ and $\timestamp_{\max},$ while $\Pi_{\byzantineagreement}$ ensures termination within expected constant time.

\begin{lemma} \label{lemma:TA-complexity}
    If all honest nodes hold inputs $\timestamp_{\inputt},$ then honest nodes output within expected  $\mathcal{O}(\log(\timestamp_{\max} - \timestamp_{\min}))$ rounds, where $\timestamp_{\min}$  and $\timestamp_{\max}$ denote the lowest and the highest honest inputs respectively (hence $\mathcal{O}(\log \Delta_{\ext})$ rounds if the synchrony assumptions are satisfied).
\end{lemma}

\section{Analysis of the Main Protocol}\label{section:main-protocol-proofs}

We now formally prove the properties of the transaction submission protocol. In particular, we prove Theorem \ref{theorem:main-protocol}.

\begin{lemma}[Honest-User Liveness]
    If a transaction $\tx$ is sent by an honest user, it gets processed and submitted to the mempool eventually, and, if the user's messages reach the nodes within $\Delta_{\ext}$ time and the synchrony assumptions hold inside the DCN for an additional $\Delta_{\dcn}$ time, the transaction get submitted within expected $\mathcal{O}(\log \Delta_{\ext})$ communication rounds.
\end{lemma}
\begin{proof}
    Since $\tx$ was sent by an honest user, all honest nodes receive the necessary messages to join $\Pi_{\timestampagreement},$ and hence they obtain a timestamp $\tau.$ Then all honest nodes obtain $\tau$ and send their shares to the other nodes. Since $f + 1$ shares are necessary to reconstruct $\tx$ and the shares are signed by the user (therefore the corrupted nodes cannot send corrupted shares), the honest nodes are able to reconstruct the transaction and submit it to the mempool.
    The round complexity follows from Lemma \ref{lemma:TA-complexity}, and from the fact that the main protocol only adds a constant number of communication rounds over $\Pi_{\timestampagreement}.$
\end{proof}

\begin{lemma}[Integrity]
    If a transaction $\tx$ gets submitted to the mempool, the process was initiated by some user.
\end{lemma}
\begin{proof}
    Submitting a transaction to the mempool requires signatures from $f + 1$ nodes, hence from at least one  honest node. This honest node only signs if it has obtained output in the invocation of $\Pi_\timestampagreement$ corresponding to the hash of $\tx.$ This means that honest nodes have joined this execution of $\Pi_{\timestampagreement},$ hence they have received input from some user.
\end{proof}

\begin{lemma}[Unique Timestamp]
    If a transaction $\tx$ gets submitted to the mempool with timestamps $\timestamp$ and $\timestamp',$ then $\timestamp = \timestamp'.$
\end{lemma}
\begin{proof}
    Assume that $\timestamp \neq \timestamp'.$ First, note that timestamps are obtained via $\Pi_\timestampagreement,$ which assigns $\tx$ a unique timestamp by the Agreement property. Therefore, during an invocation of the main protocol for $\tx,$ all honest nodes obtain the same timestamp $\tau.$ Since $f + 1$ signatures are required for the transaction to be submitted along with its timestamp, the corrupted parties are unable to submit $\tx$ to the mempool on their own. Hence, if $\timestamp \neq \timestamp',$ there must be an honest party that has signed $\timestamp',$ which happened through a different invocation of the main protocol, hence for a different transaction (ensured by the transaction's nonce).
\end{proof}

\begin{lemma}[Fair Timestamp]
    If a transaction $\tx$ gets submitted to the mempool with timestamp $\timestamp,$ then $\timestamp$ is a \emph{fair} timestamp.
\end{lemma}
\begin{proof}
    Since $\tx$ was assigned a timestamp obtained via $\Pi_\timestampagreement,$ all honest parties have assigned a fair timestamp $\tau$ to $\tx,$ i.e., satisfying $f$-Median Validity or $\lceil f/2 \rceil$-Median Validity, depending on the network conditions and on the user's honesty. Then, since $f + 1$ signatures are required for $\tx$ to be submitted, its timestamp was signed by an honest party, hence it is fair.
\end{proof}

\section{Discussion}
\textbf{Front-running Resistance.}
The DCN effectively prevents tolerant front-running, i.e., the attacker's transaction executing before the victim's transaction. Transactions are ordered according to the timestamps returned by the DCN, which fulfill $\delta$-Median Validity. Still, transactions submitted close to each other in time could receive the same timestamp, in which case the validator picks an order, or receives timestamps in the opposite order of the actual submission times. However, this is not an issue, as the transaction contents are hidden until the timestamp is agreed upon by the nodes. Thus, tolerant front-running, which, to be effective, requires the attacker to know the contents of the victim's transaction, is prevented. 

The DCN does not address destructive front-running. To be more easily integrateable in the current blockchain infrastructure, the permissioned DCN only supplies the timestamp but does not interfere with the blockchain's consensus. Destructive front-running, thus, remains possible, as the block proposer (miner) could choose not to include the transaction.  

\textbf{Censorship Resistance.} We note that by using timestamps as a decision factor when including transactions in a block, transactions become in some sense \emph{block-bound}. Thus, a transaction can become temporarily censored if the block proposer does not include the transaction and the transaction's timestamp is too low to be included in future blocks. Further,
under an asynchronous network, transactions may get lost solely due to messages getting delayed and the corresponding timestamp becoming obsolete by the time messages reach the validators. In a real-world implementation of our system, this issue can be circumvented by allowing users to resubmit their transactions with new nonces if they get lost. Importantly, the DCN does not decrease the censorship resilience any further than the block-bound model does in comparison to the classical non-block-bound blockchain model, while having the advantage of providing guarantees against front-running, which neither the classical nor the block-bound model can achieve alone. This follows both under synchrony and asynchrony by the Honest-User Liveness property.

\textbf{Permissioned Network.} The DCN is designed as a permissioned network consisting of specialized parties offering efficient and reliable transaction timestamping. Importantly, the DCN only supplies transactions with timestamps and is therefore designed to be used together with an existing permissionless blockchain. In particular, the responsibility of adding blocks to the ledger, validating blocks and storing the blockchain itself remains in the hand of the permissionless set of miners or validators. Essentially, the permissioned nature of the DCN does not reduce the robustness and decentralization of the network of validators that verify the blocks, i.e., proves that they are honest. The permissionless network, thus, retains control of the most fundamental task.

Similar permissioned setups are already common in practice today. For instance, Chainlink oracles bringing price data from the real world onto the blockchain usually operate in a similar fashion~\cite{defillamaoracles}. Moreover, since Ethereum's transition from Proof-of-Work to Proof-of-Stake, block building has become more concentrated~\cite{yang2022sok,wahrstatter2023time,heimbach2023ethereum}, in that currently more than 90\% of the blocks are built with \textit{proposer-builder separation (PBS)}~\cite{pbs2023}. In the first six months since the merge, a mere 133 builders have built these PBS blocks that were included on the ledger~\cite{heimbach2023ethereum}. With PBS, block building is no longer done by the validators themselves but is instead handled by highly sophisticated block builders~\cite{pbs2023}, similarly in spirit to how the DCN is used for timestamping transactions. By shifting tasks requiring a high degree of complexity away from validators, such as building blocks, or timestamping transactions in the case of the DCN, the requirements to run a validator node decrease. Consequently, in the long run, the number of validators is expected to increase, leading to a higher overall degree of decentralization of the consensus layer~\cite{posvspow,pbs2023}, i.e., the core of the blockchain. In our case, parties participating in the DCN are now responsible for the non-trivial task of ordering transactions, while the complexity for the validators decreases. In particular, the task of block building becomes easier as validators must simply order transactions according to their timestamp.

Finally, we note that PBS and the DCN are incompatible. While the former optimizes for block value and thereby likely includes front-running transactions, the latter is designed to achieve a fair ordering that prevents front-running. If the DCN were integrated into a permissionless blockchain instead of PBS, the blockchain would protect users from front-running as opposed to maximizing block value on their behalf.

\section{Conclusion and Future Work}
We introduced the DCN, a novel and practical solution for fair transaction ordering in permissionless blockchains. Our approach differs from previous works by treating fair ordering as a Byzantine Agreement problem rather than a Byzantine State Machine Replication problem, leading to a simpler and faster algorithm while achieving good fairness guarantees. In particular, our new timestamp agreement protocol achieves $\ceil*{f / 2}$-Median Fairness when the network is synchronous and falls back to a guarantee of $f$-Median Fairness during periods of asynchrony. These two bounds are the best that can be obtained in terms of $\delta$-Median Fairness for the synchronous and asynchronous cases, respectively, as we have shown. The asynchronous fallback paradigm is a relatively unexplored, yet more robust notion than partial synchrony, so we find it natural to use it in designing other blockchain network protocols under realistic conditions. 

As a next step, it would be valuable to consider the implementation of a dynamic set of nodes in the DCN, supporting the addition and removal of nodes in a controlled manner and the updating of related information. To do so, it will also be important to provide incentives for the nodes. One possible way to do so is to use rewards coming from transaction fees, similar to gas fees in other blockchain systems. Additionally, it would be of interest to develop a prototype of our proposed method and evaluate its performance on-chain. Finally, the DeFi scene would benefit from the development of approaches for combating destructive front-running, which our work does not address.

\bibliography{references}

\end{document}

%% file: imports/notations.tex
\newcommand{\clock}{v}

\newcommand{\timestamp}{\tau}

\newcommand{\ext}{\textsf{EXT}}
\newcommand{\dcn}{\textsf{DCN}}

\newcommand{\inputt}{\textsf{in}}
\newcommand{\outputt}{\textsf{out}}

\newcommand{\startt}{\textsf{start}}

\newcommand{\init}{\textsf{init}}
\newcommand{\approxagreement}{\textsf{AA}}
\newcommand{\byzantineagreement}{\textsf{aBA}}
\newcommand{\timestampagreement}{\textsf{TA}} 

\DeclarePairedDelimiter\abs{\big\lvert}{\big\rvert}